%% file: adaptive-influence.tex
\documentclass[11pt]{article}

\input{preamble}

\usepackage{enumitem}
\renewcommand{\epsilon}{\varepsilon}
\begin{document}
\title{Adaptive Boolean Monotonicity Testing in Total Influence Time}
\date{}

\author{Deeparnab  Chakrabarty \\
 Dartmouth College\\
{\tt deeparnab@dartmouth.edu}
\and
C. Seshadhri \\
University of California, Santa Cruz\\
{\tt sesh@ucsc.edu}
}

\def\hf{\hat{f}}
\def\hg{\hat{g}}
\def\Inf{{\sf Inf}}
\def\Viol{{\sf Viol}}
\def\I{{\mathbf I}}
\def\V{{\mathbf V}}
\def\R{{\mathbf R}}
\def\bb{\mathbf{b}}
\def\bI{\mathbf{I}}

\maketitle
\begin{abstract} The problem of testing monotonicity 
of a Boolean function $f:\{0,1\}^n \to \{0,1\}$ has received much attention
recently. Denoting the proximity parameter by $\varepsilon$, the best tester is the non-adaptive $\widetilde{O}(\sqrt{n}/\eps^2)$ tester 
of Khot-Minzer-Safra (FOCS 2015). Let $\bI(f)$ denote the total influence
of $f$. We give an adaptive tester whose running time is $\bI(f)\poly(\varepsilon^{-1}\log n)$.
\end{abstract}

\section{Introduction} \label{sec:intro}

Consider the boolean hypercube $\{0,1\}^n$, endowed with the coordinate-wise partial order denoted by $\prec$.
A function $f:\{0,1\}^n \to \{0,1\}$ is monotone if $\forall x \prec y, f(x) \leq f(y)$. 
Let $d(f,g)$, the distance between two functions $f, g$, be $|\{x: f(x) \neq g(x)\}|/2^n$.
The distance to monotonicity, $\eps_f$, is the minimum distance of $f$ to a monotone $g$.
The problem of monotonicity testing is as follows. Given query access to $f$ and a proximity parameter $\eps > 0$,
design a (randomized) procedure that 
accepts if $f$ is monotone ($\eps_f = 0$), and rejects if $f$ is $\eps$-far from monotone ($\eps_f > \eps$).
Both the above guarantees hold with probability at least $2/3$.
Such a procedure is called a monotonicity tester. A tester is non-adaptive if the queries made
are independent of the values of $f$, and adaptive otherwise. A tester is one-sided
if it accepts monotone functions with probability $1$, and two-sided otherwise.

The problem of monotonicity testing over the hypercube has been the subject of much study~\cite{Ras99,GGLRS00,DGLRRS99,ChSe13-j,ChenST14,ChenDST15,KMS15,BeBl16,Chen17}.
The first result was the $O(n/\eps)$ tester by Goldreich et al. and Dodis et al.~\cite{GGLRS00,DGLRRS99} (refer to Raskhodnikova's thesis~\cite{Ras99}).
The authors unearthed a connection to directed isoperimetric theorems to give a $o(n)$ 
tester~\cite{ChSe13-j}, whose analysis was further improved by Chen et al.~\cite{ChenST14}. In a remarkable result,
Khot-Minzer-Safra (henceforth KMS) designed an $\otilde(\sqrt{n}/\eps^2)$-query 
non-adaptive, one-sided tester~\cite{KMS15}.
The key ingredient is a directed analog of Talagrand's isoperimetry theorem~\cite{Tal93}.

The significance of the $\sqrt{n}$ bound of KMS to $f$ is underscored by a nearly
matching lower bound by Fischer et al. for non-adaptive, one-sided testers~\cite{FLNRRS02}.
A nearly matching two-sided, non-adaptive lower bound of $\Omega(n^{1/2-\delta})$ for any $\delta > 0$ was showed by Chen et al.~\cite{ChenST14,ChenDST15} using highly non-trivial techniques of generalized Central Limit Theorems.

In a major advance, Belovs and Blais~\cite{BeBl16} proved the first polynomial \emph{adaptive} (two-sided) lower bound 
of $\widetilde{\Omega}(n^{1/4})$. This was further improved by Chen et al.~\cite{Chen17} to $\widetilde{\Omega}(n^{1/3})$.
Both of these are highly non-trivial results, showing that challenges in arguing about
adaptive monotonicity testers. Interestingly, Belovs and Blais~\cite{BeBl16} show $O(\log n) + 2^{O\left(1/\eps^3\right)}$ time adaptive monotonicity testers
for Regular Linear Threshold Functions, which form the hard distribution of the nearly $\Omega(\sqrt{n})$ lower bound
of Chen et al.~\cite{ChenST14,ChenDST15}.

\noindent
This leads to the main open problem in monotonicity testing of Boolean functions: Does adaptivity always help?

\subsection{Results}

Given our (current) inability to resolve that question, we ask a refined, simpler question. 
Is there some natural condition of functions under which the $\sqrt{n}$ non-adaptive complexity of KMS can be beaten?

Let $\bI(f)$ denote the total influence (also called the average sensitivity) of $f$. 
Letting $\cD$ be the uniform distribution
over all pairs $(x,y)$ at Hamming distance $1$, $\bI(f) := n\cdot\Pr_{(x,y) \sim \cD}[f(x) \neq f(y)]$.
Our main theorem follows.

\begin{theorem} \label{thm:intro} Consider the class of functions with total influence at most $I$.
There exists a one-sided, adaptive monotonicity tester for this class with running time $O(I\cdot \poly(\log n)/\eps^4)$.
\end{theorem}

Thus, for low influence functions, the KMS $\sqrt{n}$ bound can be (adaptively) beaten.
The non-adaptive lower bound of Fischer et al.~\cite{FLNRRS02} only requires constant influence functions,
so \Thm{intro} shows that adaptivity provably helps for such functions. 
For $\bI(f) \gg \sqrt{n}$, a claim of KMS (explained later) implies the existence of non-adaptive $O(n/\bI(f)))$ testers.
The tradeoff between \Thm{intro} and this bound is basically $\sqrt{n}$,
the maximum possible influence of a monotone function.
We note that all adaptive lower bound constructions 
have functions with influence $\Theta(\sqrt{n})$~\cite{BeBl16,Chen17}.

\section{Preliminaries} \label{sec:prelims}

We use $H_n$ to denote the standard (undirected) hypercube graph on $\{0,1\}^n$, where all pairs at Hamming distance $1$
are connected. An edge $(x,y)$ of $H_n$ is influential, if $f(x) \neq f(y)$. The number of
influential edges is precisely $\bI(f)\cdot 2^{n-1}$.  An influential
edge $(x,y)$ is a violating edge if $x \prec y$, $f(x) = 1$, and $f(y) = 0$. Our tester
will perform random walks of $H_n$. Note that $H_n$ is regular, so this is a symmetric Markov Chain.

We crucially use the central result of KMS, in essence a directed analogue of Talagrand's isoperimetric
theorem.
% Note that the maximum
% influence of a monotone function is $\Theta(\sqrt{n})$, 

\begin{lemma}(Lemma 7.1 in ~\cite{KMS15}, paraphrased)\label{lem:kms}
Given any Boolean function $f:\{0,1\}^n \to \{0,1\}$ which is $\epsilon$-far from being monotone, 
there exists a subgraph $G = (A,B,E)$ of the hypercube and parameters $\sigma \in (0,1)$, $d\in \mathbb{Z}_{>0}$ such that
\begin{itemize}[noitemsep]
\item Each edge $(a,b)\in E$ with $a\in A$ and $b\in B$ is a violating edge.
\item $|B| = \sigma\cdot 2^n$.
\item The degree of each vertex in $B$ is exactly $d$ and the degree of each vertex in $A$ is at most $2d$.
\item $\sigma^2d = \Theta(\epsilon^2/\log^4 n)$.
\end{itemize}
\end{lemma}

We also mention a simpler lemma from KMS which we use. % that can be obtained by some Fourier coefficient manipulations.

\begin{lemma}(Theorem 9.1 in~\cite{KMS15}, paraphrased) \label{lem:inf}
If $\bI(f) > 6\sqrt{n}$, then a constant fraction of influential edges are violating edges.
\end{lemma}

When $\bI(f) > 6\sqrt{n}$, one can find a violation in time $O(n/\bI(f)) = O(\bI(f))$ by simply
sampling random edges. Thus, we will focus on the case $\bI(f) \leq 6\sqrt{n}$.

\paragraph{Idea.} Our analysis is short and follows the analysis of the KMS tester. The tester
of KMS can be throught of as performing random walks on the \emph{directed} hypercube with orientation corresponding to the partial order.
Their analysis lower bounds the probability of encountering a violation.
Our insight is that one can perform an analogous analysis for walks on the \emph{undirected}
hypercube $H_n$. Suppose we perform an $\ell$-step random walk (on $H_n)$ from $x$ that ends at $y$.
If $f(x) \neq f(y)$, then the walk clearly passed through an influential edge.
The power of adaptivity is that we can find such an influential edge through binary search.
This idea of using binary search is also present in the algorithm of Belovs and Blais~\cite{BeBl16} 
to adaptively test regular linear threshold functions.
We bound the probability that this influential edge is a violation.
Our insight is that, by setting the length $\ell$ appropriately,
this probability can be lower bounded by (essentially) $1/\bI(f)$.

\section{Tester}
Let $f:\{0,1\}^n \to \{0,1\}$ be a Boolean function over the hypercube with total influence $\bI(f)$.

\begin{figure}[h]
	\begin{framed}
		\noindent \textbf{Input:} A Boolean function $f: \{0,1\}^n \to \{0,1\}$ and a parameter $\eps \in (0,1)$
		\begin{enumerate}[noitemsep]
			\item Choose $k\in_R \{0,1,2,\ldots,\ceil{\log n}\}$ uniformly at random. Set $\ell := 2^k$.
			\item Choose $x\in \{0,1\}^n$ uniformly at random.
			\item Perform an $\ell$-length random walk $p$ on $H_n$ starting from $x$ to reach $y\in \{0,1\}^n$.
			\item If $f(x) \neq f(y)$:
			\begin{enumerate}
				\item Perform binary search on $p$ to find an influential edge $(u,v)\in p$.
				\item REJECT if $(u,v)$ is a monotonicity violation.
			\end{enumerate} 
            \item If not REJECTED so far, then ACCEPT.
			
		\end{enumerate} 
	\end{framed}
	\caption{\small{\textbf{Adaptive Monotonicity Tester for Boolean Functions}}}
	\label{fig:alg}
\end{figure}

\noindent
\Thm{intro} follows directly from the following.
\begin{theorem}\label{thm:main}
Assume $\bI(f) \leq 6\sqrt{n}$.
If $f$ is $\epsilon$-far from being monotone, then Algorithm~\ref{fig:alg} rejects with probability $\Omega\left(\frac{\epsilon^4}{\bI(f)\log^9 n}\right)$.
\end{theorem}

\begin{definition}\label{def:pers}
Given a positive integer $\ell$, we call a vertex $x\in \{0,1\}^n$ {\em $\ell$-sticky} if 
an $\ell$-length random walk from $x$ on $H_n$ contains no influential edges with probability $\geq 1/2$.
A vertex is called non-$\ell$-sticky otherwise.
An edge is $\ell$-sticky if both endpoints are $\ell$-sticky.
\end{definition}
\noindent
The following simple monotonicity property follows from the definition.
\begin{observation} \label{obs:sticky}
	If $x$ is $\ell$-sticky and $\ell' < \ell$, then $x$ is $\ell'$-sticky as well.
\end{observation}

%Since $H_n$ is regular, the distribution of $\ell$-length random walks (starting from a uniform random vertex) is simply the uniform
%distribution over all $\ell$-length walks.

\begin{lemma}\label{lem:simple}
The fraction of non-$\ell$-sticky vertices of a hypercube is at most $\frac{2\ell\cdot \bI(f)}{n}$.
\end{lemma}
\begin{proof}
Given $x\in \{0,1\}^n$ and a positive integer $\ell > 0$, define the random variable
$Z_{x,\ell}$ that is the number of influential edges in a random walk of length $\ell$ starting from $x$. 
Therefore, $x$ is $\ell$-non-sticky iff $\Pr[Z_{x,\ell} > 0] > 1/2$.
Let $N$ denote the set of $\ell$-non-sticky vertices.

%By definition, we get $\Pr_y[f(x) \neq f(y)] \leq \Pr[Z_{x,\ell} > 0]$.
Since $Z_{x,\ell}$ is non-negative and integer valued we get $\Pr[Z_{x,\ell} > 0] \leq \Exp[Z_{x,\ell}]$.
\begin{equation}\label{eq:simple}
|N|/2^n < \frac{2}{2^n} \sum_{x \in N} \Pr[Z_{x,\ell} > 0] < \frac{2}{2^n}\sum_{x\in \{0,1\}^n}\Pr[Z_{x,\ell}>0] ~~\leq~~ \frac{2}{2^n}\sum_{x\in \{0,1\}^n} \Exp[Z_{x,\ell}]
\end{equation}
The RHS above is precisely twice the expected number of influential edges encountered in an $\ell$-length
random walk starting from the uniform distribution on $H_n$. Let $\cP_\ell$ denote the uniform distribution
on $\ell$-length paths in $H_n$. For $p \sim \cP_\ell$, $p_t$ denotes the $t$th edge in $p$,
and let $\chi(e)$ be the indicator for edge $e$ being influential.
The RHS of~\eqref{eq:simple} is equal to $2\Exp_{p \sim \cP_\ell}[\sum_{t \leq \ell} \chi(p_t)]
= 2\sum_{t \leq \ell} \Exp_{p \sim \cP_\ell} [\chi(p_t)]$.
Since the uniform distribution is stationary for random walks on $H_n$,
the distribution induced on $p_t$ is the uniform distribution on edges in $H_n$.
Thus, $\Exp_{p \sim \cP_\ell} [\chi(p_t)] = \bI(f)/n$ and the RHS of~\eqref{eq:simple}
is $2\ell \cdot \bI(f)/n$.
\end{proof}
\noindent
For any integer $\ell > 0$, let $F_\ell$ be the set of $\ell$-sticky violating edges. That is, 
\[
F_\ell := \{(x,y) \in H_n : (x,y) \textrm{ is violating and}, ~ x,y ~\textrm{ are $\ell$-sticky} \}
\]

\begin{lemma}\label{lem:simple2}
If $\ell$ is the length of the random walk chosen in Step 1, then
Algorithm~\ref{fig:alg} rejects with probability $\Omega\left(\frac{\ell}{n}\cdot \frac{|F_\ell|}{2^n}\right)$.
\end{lemma}
\begin{proof}
Fix an edge $(u,v) \in F_\ell$. Let $p = (e_1,\ldots,e_\ell)$ be the edges of a random walk $p$. 
Each edge $e_t$ for $1\leq t\leq \ell$ is a uniform random edge of $H_n$. Therefore, for any fixed edge $(u,v)$, $\Pr[e_t = (u,v)] = \frac{2}{n2^n}$.

Now, given $1\leq t\leq \ell$,  define $\calE^t_{u,v}$ to be the event that the $t$th edge of $p$ is $(u,v)$ {\em and} no other edge of $p$ is influential.
Since $(u,v)$ is itself influential, $\calE_{u,v}$ is the disjoint union $\vee_{t=1}^\ell \calE_{u,v}^t$.
Furthermore, for two distinct edges $(u,v)$ and $(u',v')$ in $F_\ell$, the events $\calE_{u,v}$
and $\calE_{u',v'}$ are disjoint.

\noindent
Observe that
\begin{equation}\label{eq:dingo}
\Pr[\textrm{Algorithm~\ref{fig:alg} rejects given $\ell$}] \geq \Pr[\bigvee_{(u,v)\in F_\ell} \calE_{u,v}] = \sum_{(u,v)\in F_\ell} \Pr[\calE_{u,v}]
\end{equation}
The equality follows since the events are mutually exclusive. The inequality follows since if $\calE_{u,v}$ occurs then the end points of $p$ must have differing values and binary search on $p$ will return the violation $(u,v)$. 

Consider the event $\calE_{u,v}^t$. 
For this event to occur,  $e_t$ must be $(u,v)$. Consider the conditional probability $\Pr[\calE_{u,v}^t~|~e_t = (u,v)]$.
%Since each edge in $p$ is a uniform edge of the hypercube, this occurs with probability $\geq 2/n2^n$.
Let $\calF_u$ be the event that a $(t-1)$-length random walk from $u$ contains no influential edges, and let $\calF_v$ be the event that an {\em independent} $(\ell-t)$-length random walk from $v$ contains no influential edges. 
\begin{claim}\label{clm:sigh}
	$\Pr[\calE_{u,v}^t~|~e_t = (u,v)] = \Pr[\calF_u \wedge \calF_v] = \Pr[\calF_u]\cdot\Pr[\calF_v]$
\end{claim}

\begin{proof} Conditioned on $e_t = (u,v)$, the distribution of the first $(t-1)$ steps of the random walk
is the uniform distribution of $(t-1)$-length paths that end at $u$. This is the same distribution of
the $(t-1)$-length random walk starting from $u$.
The distribution of the last $(\ell-t)$ steps, conditioned on $e_t = (u,v)$, is the $(\ell-t)$-length
random walk starting from $v$.
\end{proof}
\noindent
Since $(u,v)$ is an $\ell$-sticky edge, by \Obs{sticky} and \Def{pers},
$\Pr[\calF_u] \geq 1/2$ and $\Pr[\calF_v] \geq 1/2$. 
The proof is completed by plugging the following bound into 
\eqref{eq:dingo}.
\begin{align*}
\Pr[\calE_{u,v}] =  \sum_{t=1}^\ell \Pr[\calE_{u,v}^t] 
=  \sum_{t=1}^\ell \Pr[e_t = (u,v)] \cdot \Pr[\calE_{u,v}^t~|~ e_t = (u,v)]
& = \frac{2}{n2^n} \sum_{t=1}^\ell \Pr[\calE_{u,v}^t~|~ e_t = (u,v)] \\
& \geq \frac{\ell}{4n2^n}
\end{align*}
\end{proof}
\noindent
We complete the proof of \Thm{main}.
\begin{proof}[Proof of Theorem~\ref{thm:main}]
Let $G = (A,B,E), \sigma, d$ be as in \Lem{kms}. 

First, we take care of a corner case.
Suppose $\sigma < 100/\sqrt{n}$. Since $\sigma^2d = \Omega(\eps^2/\log^4 n)$, 
we get  $\sigma d = \Omega(\eps^2 \sqrt{n}/\log^4n)$.
Since the number of violating edges is $\geq |E| = \sigma d 2^n$ and since each 
violating edge is influential, we get that $\bI(f) \geq |E|/2^n =  \Omega(\eps^2\sqrt{n}/\log^4n)$.
With probability $1/\log n$, Algorithm~\ref{fig:alg} sets $\ell = 1$, in which case
the tester basically checks whether a uniform random edge in $H_n$ is a violation.
Conditioned on setting $\ell = 1$,
the rejection probability is $\Omega(2|E|/n2^n) = \Omega(\eps^2/\sqrt{n}\log^4n) = \Omega(\eps_f^4/\bI(f)\log^8 n)$.
Thus, in the case $\sigma < 100/\sqrt{n}$, we have proved \Thm{main}.\smallskip

\noindent
Henceforth,  assume that $\sigma \geq 100/\sqrt{n}$.
Since $\bI(f) \leq 6\sqrt{n}$,
we deduce that $\frac{n\sigma}{16\cdot\bI(f)} > 1$.
Thus, there exists a non-negative power of $2$ (call it $\ell^*$) 
such that $\frac{\sigma}{16} < \frac{\ell^*\cdot \bI(f)}{n} \leq \frac{\sigma}{8}$. 

Let $A'$ and $B'$ be the subset of $A$ and $B$ that are $\ell^*$-sticky. Let $E'\subseteq E$ be the edges with end points in $A'$ and $B'$.
Note that any edge of $E' \subseteq F_{\ell^*}$. By 
Lemma~\ref{lem:simple}, we get that the fraction of $\ell^*$-non-sticky nodes is at most $2\ell^*\cdot \bI(f)/n \leq \sigma/4$.
Since the degree of any node in $G$ in Lemma~\ref{lem:kms} is $\leq 2d$, we get 
\[
|F_{\ell^*}| \geq |E'| \geq |E| - (2d)\cdot \frac{\sigma\cdot 2^n}{4} = \frac{\sigma d 2^n}{2}. 
\]

The probability the algorithm chooses $\ell = \ell^*$ is $1/{\log n}$. Lemma~\ref{lem:simple2} gives us
\begin{align*}
\Pr[\textrm{Algorithm rejects}]  & \geq  \frac{1}{\log n}\cdot \Pr[\textrm{Algorithm rejects}| \ell = \ell^*] \\
& \geq \frac{1}{\log n} \cdot \left(\frac{\ell^*}{n}\cdot \frac{|F_{\ell^*}|}{2^n}\right)  ~~~~~~~ (\textrm{by Lemma~\ref{lem:simple2}}) \\
& \geq \frac{1}{\log n} \cdot \frac{\ell^*}{n} \cdot \frac{\sigma d}{2} \\
& \geq \frac{1}{\bI(f)} \cdot \left(\frac{\sigma^2 d}{32\log n}\right) ~~~~~~~ (\textrm{plugging $\ell^* \geq \sigma n/(16\cdot\bI(f))$})\\
& \geq \frac{1}{\bI(f)}\cdot \frac{\epsilon^2}{\log^5 n} ~~~~~~~ (\textrm{by Lemma~\ref{lem:kms}})
\end{align*}
\end{proof}

\bibliographystyle{alpha}
\bibliography{derivative-testing}

\end{document}

%% file: preamble.tex
	\usepackage{a4,geometry}
\usepackage{graphicx}
\usepackage{amsmath,amssymb,amsthm,mathtools}
\usepackage{paralist}
\usepackage{bm}
\usepackage{xspace}
\usepackage{url}
\usepackage{fullpage, prettyref}
\usepackage{boxedminipage}
\usepackage{wrapfig}
\usepackage{ifthen}
\usepackage{color}
\usepackage{xcolor}
\usepackage{framed}
\usepackage{algorithmic,algorithm}
\usepackage[pagebackref,letterpaper=true,colorlinks=true,pdfpagemode=none,urlcolor=blue,linkcolor=blue,citecolor=violet,pdfstartview=FitH]{hyperref}
\usepackage{fullpage}

\usepackage{thmtools}
\usepackage{thm-restate}
\newtheorem{theorem}{Theorem}[section]

\newtheorem{lemma}[theorem]{Lemma}
\newtheorem{claim}[theorem]{Claim}

\newtheorem{definition}[theorem]{Definition}

\newtheorem{observation}[theorem]{Observation}

\newcommand{\ignore}[1]{}

\newcommand{\cD}{\mathcal{D}}

\newcommand{\cP}{\mathcal{P}}

\newcommand{\R}{\mathbb R}

\newcommand{\eps}{\varepsilon}

\newcommand{\poly}{\mathrm{poly}}

\newcommand{\calE}{{\cal E}}

\newcommand{\calF}{{\cal F}}

\newcommand{\bb}{\boldsymbol{b}}

\newcommand{\ceil}[1]{\lceil#1\rceil}
\newcommand{\Exp}{\EX}

\newcommand{\EX}{\hbox{\bf E}}

\newcommand{\otilde}{\widetilde{O}}

\newcommand{\Sec}[1]{\hyperref[sec:#1]{\S\ref*{sec:#1}}} %section
\newcommand{\Eqn}[1]{\hyperref[eq:#1]{(\ref*{eq:#1})}} %equation
\newcommand{\Fig}[1]{\hyperref[fig:#1]{Fig.\,\ref*{fig:#1}}} %figure
\newcommand{\Tab}[1]{\hyperref[tab:#1]{Tab.\,\ref*{tab:#1}}} %table
\newcommand{\Thm}[1]{\hyperref[thm:#1]{Theorem\,\ref*{thm:#1}}} %theorem
\newcommand{\Fact}[1]{\hyperref[fact:#1]{Fact\,\ref*{fact:#1}}} %fact
\newcommand{\Lem}[1]{\hyperref[lem:#1]{Lemma\,\ref*{lem:#1}}} %lemma
\newcommand{\Prop}[1]{\hyperref[prop:#1]{Prop.~\ref*{prop:#1}}} %property
\newcommand{\Cor}[1]{\hyperref[cor:#1]{Corollary~\ref*{cor:#1}}} %corollary
\newcommand{\Conj}[1]{\hyperref[conj:#1]{Conjecture~\ref*{conj:#1}}} %conjecture
\newcommand{\Def}[1]{\hyperref[def:#1]{Definition~\ref*{def:#1}}} %definition
\newcommand{\Alg}[1]{\hyperref[alg:#1]{Alg.~\ref*{alg:#1}}} %algorithm
\newcommand{\Ex}[1]{\hyperref[ex:#1]{Ex.~\ref*{ex:#1}}} %example
\newcommand{\Clm}[1]{\hyperref[clm:#1]{Claim~\ref*{clm:#1}}} %example
\newcommand{\Step}[1]{\hyperref[step:#1]{Step~\ref*{step:#1}}} %example
\newcommand{\Obs}[1]{\hyperref[obs:#1]{Obs~\ref*{obs:#1}}} %example